\documentclass[aip,amsmath,amssymb,reprint]{revtex4-1}

\usepackage{graphicx}
\usepackage{dcolumn}
\usepackage{bm}

\usepackage[utf8]{inputenc}
\usepackage[T1]{fontenc}
\usepackage{mathptmx}
\usepackage{etoolbox}

\makeatletter
\def\@email#1#2{
 \endgroup
 \patchcmd{\titleblock@produce}
  {\frontmatter@RRAPformat}
  {\frontmatter@RRAPformat{\produce@RRAP{*#1\href{mailto:#2}{#2}}}\frontmatter@RRAPformat}
  {}{}
}
\makeatother

\usepackage{hyperref}
\usepackage{amsbsy}
\usepackage{amsfonts}
\usepackage{amsthm}
\usepackage{physics}
\usepackage{dsfont}
\usepackage{multirow}
\usepackage{mathtools}

\newcommand{\young}[2]{{\mathbb{Y}^{#1}_{#2}}}
\newcommand{\1}{\mathds{1}}
\newcommand{\SU}{\mathrm{SU}}
\newcommand{\CC}{\mathbb{C}}
\newcommand{\RR}{\mathbb{R}}

\newcommand{\eig}{\mathrm{eig}}

\newcommand{\mcA}{\mathcal{A}}
\newcommand{\mcB}{\mathcal{B}}
\newcommand{\mcU}{\mathcal{U}}
\newcommand{\mcI}{\mathcal{I}}
\newcommand{\mcO}{\mathcal{O}}
\newcommand{\mcL}{\mathcal{L}}

\newcommand{\mcH}{\mathcal{H}}

\newcommand{\mcT}{\mathcal{T}}

\newcommand{\mfT}{\mathsf{T}}

\newcommand{\dbraket}[2]{\langle\!\langle {#1} \vert {#2} \rangle\!\rangle}
\newcommand{\dketbra}[2]{\vert {#1} \rangle\!\rangle\!\langle\!\langle {#2} \vert}

\makeatletter
\def\dbraket#1{%
    \@ifnextchar\bgroup{%
        \dbraket@{#1}%
    }{%
        \langle\!\langle {#1} \vert {#1} \rangle\!\rangle%
    }%
}
\def\dbraket@#1#2{%
    \langle\!\langle {#1} \vert {#2} \rangle\!\rangle%
}
\makeatother

\makeatletter
\def\dketbra#1{%
    \@ifnextchar\bgroup{%
        \dketbra@{#1}%
    }{%
        \vert {#1} \rangle\!\rangle\!\langle\!\langle {#1} \vert%
    }%
}
\def\dketbra@#1#2{%
    \vert {#1} \rangle\!\rangle\!\langle\!\langle {#2} \vert%
}
\makeatother

\newtheorem{definition}{Definition}[section]
\newtheorem{theorem}[definition]{Theorem}
\newtheorem{corollary}[definition]{Corollary}
\newtheorem{lemma}[definition]{Lemma}
\newtheorem{conjecture}[definition]{Conjecture}

\begin{document}

\preprint{AIP/123-QED}

\title[Asymptotically optimal unitary estimation in $\SU(3)$ by the analysis of graph Laplacian]{Asymptotically optimal unitary estimation in $\SU(3)$ by the analysis of graph Laplacian}
\author{Satoshi Yoshida}
\email{satoshiyoshida.phys@gmail.com}
\affiliation{Department of Physics, Graduate School of Science, The University of Tokyo, Hongo 7-3-1, Bunkyo-ku, Tokyo 113-0033, Japan}
\author{Hironobu Yoshida}
\affiliation{Nonequilibrium Quantum Statistical Mechanics RIKEN Hakubi Research Team, Pioneering Research Institute (PRI), RIKEN, 2-1 Hirosawa, Wako, Saitama 351-0198, Japan}
\author{Mio Murao}
\affiliation{Department of Physics, Graduate School of Science, The University of Tokyo, Hongo 7-3-1, Bunkyo-ku, Tokyo 113-0033, Japan}
\affiliation{Trans-scale Quantum Science Institute, The University of Tokyo, Bunkyo-ku, Tokyo 113-0033, Japan}

\date{\today}

\begin{abstract}
Unitary estimation is the task to estimate an unknown unitary operator $U\in\SU(d)$ with $n$ queries to the corresponding unitary operation, and its accuracy is evaluated by an estimation fidelity.
We show that the optimal asymptotic fidelity of $3$-dimensional unitary estimation is given by $F_\mathrm{est}(n,d=3) = 1-{56\pi^2 \over 9n^2} + O(n^{-3})$ by the analysis of the graph Laplacian based on the finite element method.
We also show the lower bound on the fidelity of $d$-dimensional unitary estimation for an arbitrary $d$ given by $F_\mathrm{est}(n,d) \geq 1- {(d+1)(d-1)(3d-2)(3d-1) \over 6n^2} + O(n^{-3})$ achieving the best known lower bound and tight scaling with respect to $n$ and $d$.
This lower bound is derived based on the unitary estimation protocol shown in \href{https://doi.org/10.1103/PhysRevA.75.022326}{[J.~Kahn, Phys.~Rev.~A 75, 022326, 2007]}.
\end{abstract}

\maketitle

\section{Introduction}

Recent development on quantum technology has increased the demand to learn the property of quantum dynamics from the experimental data.
The dynamics of a $d$-level closed quantum system is represented by a $d$-dimensional unitary operator $U\in \SU(d)$~\cite{nielsen2010quantum}, and learning of the dynamics is modeled as the task to estimate $U$ using multiple queries to $U$.
This task is called \emph{unitary estimation}~\cite{holevo2011probabilistic, acin2001optimal, dariano2001using, fujiwara2001estimation, peres2002covariant, bagan2004entanglement, bagan2004quantum, ballester2004estimation, chiribella2004efficient, chiribella2005optimal, hayashi2006parallel, kahn2007fast, bisio2010optimal, yang2020optimal, haah2023query}, and in particular, we consider the Bayesian setting where the input unitary operator is drawn from the Haar measure of the unitary group.
One of the most important fundamental questions is the tradeoff between the \emph{query complexity}, the number of queries, and the \emph{estimation fidelity}, which expresses the accuracy of estimation.
The optimal estimation fidelity with $n$ queries, denoted by $F_\mathrm{est}(n,d)$, is shown to obey the Heisenberg limit~\cite{kahn2007fast}:
\begin{align}
    F_\mathrm{est}(n,d) = 1-{h(d)\over n^2} + O(n^{-3}),
\end{align}
for a constant $h(d)$ depending on $d$.
Since the optimal fidelity of unitary estimation is equivalent to the optimal performances of various other tasks such as parallel unitary inversion and transposition~\cite{quintino2022deterministic}, storage and retrieval of unitary operations~\cite{bisio2010optimal}, and deterministic port-based teleportation~\cite{yoshida2024one}, it is important to determine $h(d)$.
Reference~\cite{kahn2007fast} shows an upper bound on $h(d)$ in terms of the integral with numerical estimates for $d=2,3,4$, but its scaling with respect to $d$ was not known.
Reference~\cite{christandl2021asymptotic} shows $\Omega(d^2)\leq h(d)\leq O(d^5)$ for deterministic port-based teleportation, which is shown to have the same optimal fidelity as the unitary estimation in Ref.~\cite{yoshida2024one}.
Reference~\cite{yang2020optimal} shows that $h(d)\leq O(d^4)$, which is shown to be asymptotically tight [i.e., $h(d)\geq \Omega(d^4)$] in Refs.~\cite{haah2023query, yoshida2024one}.
However, the exact formula for $h(d)$ is not determined except for $d=2$, where $F_\mathrm{est}(n,d=2)$ is fully determined in Ref.~\cite{bagan2004entanglement}.

In this work, we obtain an asymptotically optimal unitary estimation of $\SU(3)$ to determine $h(d=3)$.
The derivation is based on the Dirichlet Laplacian problem.
Reference~\cite{christandl2021asymptotic} shows that $h(d)$ is upper bounded by ${\lambda_\mathrm{min}(\Omega_{d-1}) \over d}$ for deterministic port-based teleportation, where $\lambda_\mathrm{min}(\Omega_{d-1})$ is the first eigenvalue of the Dirichlet Laplacian problem on a $(d-1)$-dimensional simplex $\Omega_{d-1}$.
We show that this upper bound is tight by using the finite element method for $d=2,3$, and obtain $h(d=3)$.
This proof is based on the analysis of the graph Laplacian, which reduces the discrete Laplacian problem to the continuous one using the finite element method.
We also determine the asymptotic fidelity of the unitary estimation protocol shown in Ref.~\cite{kahn2007fast}, which is shown to be nearly tight for small $d$ and asymptotically tight for large $d$.

The rest of this paper is organized as follows.
In Sec.~\ref{sec:preliminaries}, we summarize the preliminary facts and notations required to present the main results.
In Sec.~\ref{sec:problem_setting}, we formalize the task of unitary estimation.
In Sec.~\ref{sec:main_result}, we present the main results of this work, the optimal estimation fidelity of $\SU(3)$ and nearlly optimal unitary estimation for an arbitrary dimesion $d$.
In Sec.~\ref{sec:conclusion}, we conclude this paper and discuss the future directions.
In Appendix~\ref{appendix_sec:review_FEM}, we review the finite element method.
In Appendix~\ref{appendix_sec:kahn}, we review the unitary estimation protocol shown in Ref.~\cite{kahn2007fast} and derive an analytical lower bound of $h(d)$.

\section{Preliminaries}
\label{sec:preliminaries}
\subsection{Notation and definitions}

We denote $[n]\coloneqq \{1,2,\ldots,n\}$.
For a logical expression $\varphi(x,\ldots,z)$ on variables $x,\ldots,z$, we define $\delta_{\varphi(x,\ldots,z)}$ by
\begin{align}
    \delta_{\varphi(x,\ldots,z)} \coloneqq
    \begin{cases}
        1 & (\varphi(x,\ldots,z) \text{ holds})\\
        0 & (\text{otherwise})
    \end{cases}.
\end{align}
For a set $X$, we denote the cardinality of $X$ by $\# X$.
We use the big-O notation $O(\cdot)$, $\Omega(\cdot)$ and $\Theta(\cdot)$, defined as follows~\cite{arora2009computational}:
\begin{align}
    f(x) = O(g(x)) &\Leftrightarrow \limsup_{x\to \infty}{\abs{f(x) \over g(x)}} <\infty,\\
    f(x) = \Omega(g(x)) &\Leftrightarrow g(x) = O(f(x)),\\
    f(x) = \Theta(g(x)) &\Leftrightarrow f(x)=O(g(x)) \text{ and } f(x) = \Omega(g(x)).
\end{align}
We denote the set of linear operators on a Hilbert space $\mcH$ by $\mcL(\mcH)$, and the identity operator on $\mcH$ by $\1_{\mcH}$.
For a matrix $X$, we denote the tranpose of $X$ by $X^\mfT$, and the complex conjguate of $X$ by $X^\dagger$.
For a $d$-dimensional Eucledian space $\RR^d$ with natural coordinates $(x_1,\ldots, x_d)$, we define the gradient operator $\nabla\coloneqq ({\partial \over \partial x_1},\ldots,{\partial \over \partial x_d})$ and the Laplacian $\Delta\coloneqq \nabla\cdot \nabla = \sum_{i=1}^{d} \left({\partial \over \partial x_i}\right)^2$.

The quantum state is represented by $\phi\in\mcL(\mcH)$ such that $\phi\geq 0$ and $\Tr \phi = 1$.
The measurement on a quantum state on $\mcH$ is represented by a positive operator-valued measure (POVM) $\{M_a\}_a \subset \mcL(\mcH)$ satisfying $M_a\geq 0$ and $\sum_a M_a = \1_{\mcH}$.
The probability distribution to obtain a measurement outcome $a$ is given by $\Tr(M_a \phi)$.
The transformation of a quantum state on $\mcH_1$ to a quantum state on $\mcH_2$ is called a quantum channel.
A quantum channel is represented by a completely positive and trace preserving (CPTP) map $\Lambda: \mcL(\mcH_1) \to \mcL(\mcH_2)$, i.e., $(\Lambda\otimes \1_{\mcL(\mcA)})(\phi)$ is positive for any Hilbert space $\mcA$ and positive operator $\phi\in\mcL(\mcH_1\otimes \mcA)$, and $\Tr[\Lambda(\phi)] = \Tr(\phi)$ for all $\phi\in\mcL(\mcH_1)$.
A unitary channel is a quantum channel given by
\begin{align}
    \mcU(\cdot) = U\cdot U^\dagger
\end{align}
using a unitary operator $U$.

\subsection{Representation theory of the unitary group}

We consider a representation $\rho:\SU(d) \to \mcL(\CC^d)^{\otimes n}$ of the (special) unitary group $\SU(d)$ given by
\begin{align}
    \rho(U) = U^{\otimes n} \quad \forall U\in\SU(d).
\end{align}
This representation can be decomposed into the irreducible representations as (Schur-Weyl duality~\cite{harrow2005applications})
\begin{align}
    \rho(U) \simeq \bigoplus_{\mu\in\young{d}{n}} \rho_{\mu}(U) \otimes \1_{\CC^{m_\mu}},
\end{align}
where $\young{d}{n}$ is the set of Young diagrams with $n$ boxes and at most $d$ columns, $\rho_\mu:\SU(d) \to \mcH_{\mu}$ is the irreducible representation of $\SU(d)$ corresponding to $\mu$, and $m_\mu$ is the multiplicity.
The Young diagram $\mu$ is represented by a partition of $n$ given by $d$ integers $(\mu_1,\ldots, \mu_d)$ such that $\mu_1\geq \cdots \geq \mu_d\geq 0$ and $\sum_{i=1}^{d}\mu_i = n$.
This irreducible decomposition induces an isomorphism of representation spaces:
\begin{align}
\label{eq:space_isomorphism}
    (\CC^d)^{\otimes n} \simeq \bigoplus_{\mu\in\young{d}{n}} \mcH_\mu \otimes \CC^{m_\mu}.
\end{align}

We can define an orthonormal basis in the representation space $\mcH_\mu$ called the Gelfand-Tsetlin basis $\{\ket{q_\mu}\}_{q_\mu}$.
Similarly, we can define an orthonormal basis in the multiplicity space $\CC^{m_\mu}$ called the Young-Yamanouchi basis $\{\ket{p_\mu}\}_{p_\mu}$.
By combining these two bases, we can define the Schur basis $\{\ket{q_\mu}\otimes \ket{p_\mu}\}_{\mu, q_\mu, p_\mu}$ on the left-hand side of Eq.~\eqref{eq:space_isomorphism}.
We can define the computational basis on $(\CC^{d})^{\otimes n}$ by $\{\ket{i_1}\otimes \cdots \otimes \ket{i_n}\}_{(i_1,\ldots,i_n)\in [d]^n}$ by using a natural basis $\{\ket{i}\}_{i\in[d]}$ of $\CC^d$.
The basis change from the computational basis to the Schur basis is called the quantum Schur transform, denoted by $U_\mathrm{Sch}$.

\begin{figure*}
    \centering
    \includegraphics[width=.7\linewidth]{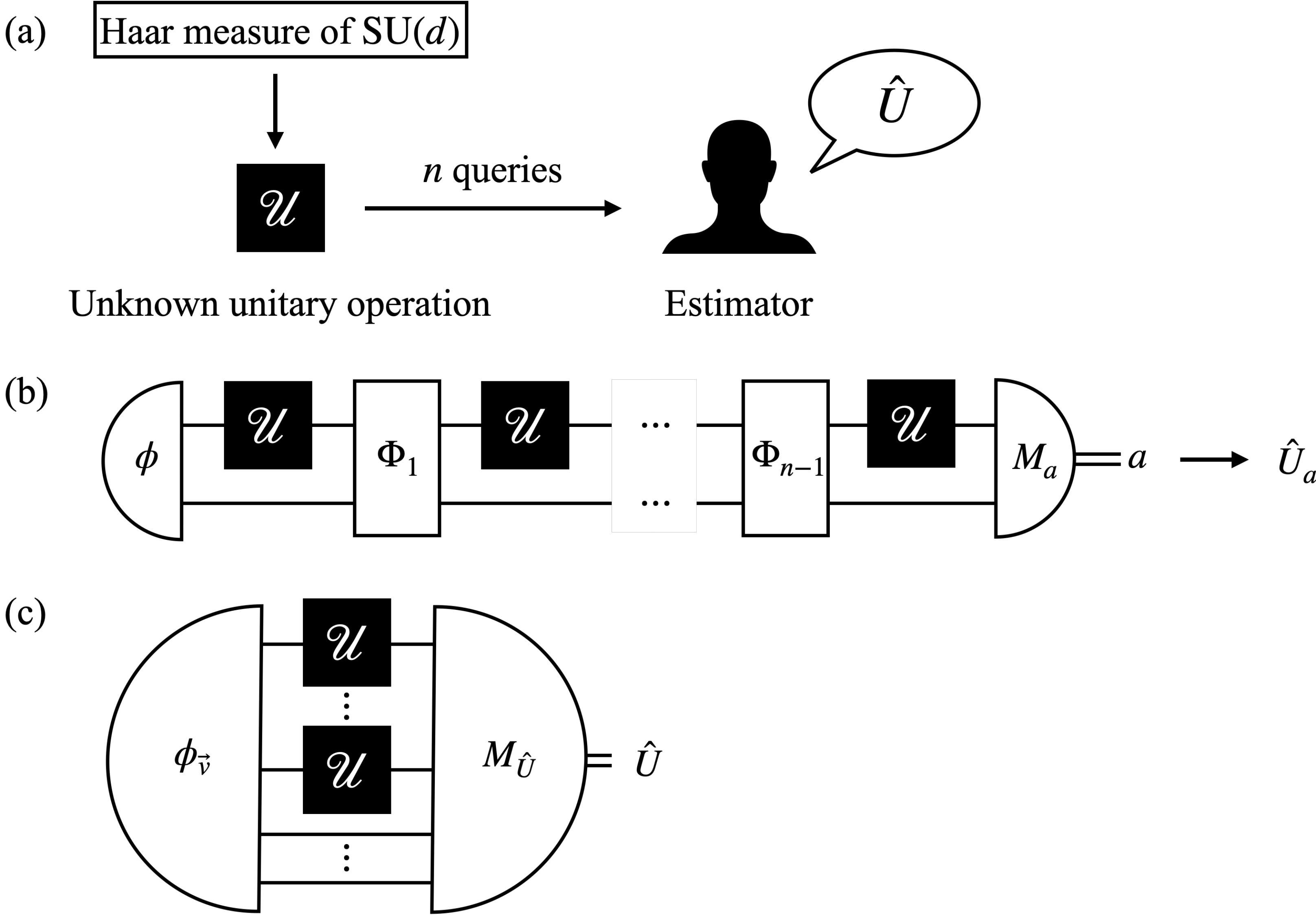}
    \caption{(a) Overview of the task of unitary estimation.
    Unknown unitary operation $\mcU$ corresponding to a unitary operator $U$ is drawn from the Haar measure of $\SU(d)$.
    The task is to estimate $U$ with $n$ queries to $\mcU$.
    (b) Most general protocol for the unitary estimation [see also Eq.~\eqref{eq:def_superinstrument}].
    Based on the measurement outcome $a$, one outputs the estimator $\hat{U}_a$.
    The optimal fidelity of unitary estimation is shown to be the covariant parallel protocol shown in (c), which outputs the estimator $\hat{U}$ as the measurement outcome.}
    \label{fig:unitary_estimation}
\end{figure*}

\section{Problem setting: Unitary estimation}
\label{sec:problem_setting}
\subsection{Definition of the task}

Suppose we are given an unknown unitary channel $\mcU$ corresponding to a unitary operator $U$ drawn from the Haar measure~\cite{mele2024introduction} of $\SU(d)$.
The task of unitary estimation is to estimate $U$ by querying $\mcU$ for $n$ times.
The most general strategy to obtain an estimator $\hat{U}$ given $n$ queries to $\mcU$ is represented by a quantum tester~\cite{chiribella2009theoretical}.
Quantum tester realizes a transformation of $n$ quantum channels $\Lambda_i: \mcL(\mcI_i) \to \mcL(\mcO_i)$ for $i\in [n]$ to a probability distribution $\{p_a\}_{a\in A}$, {where $\mcL(\mcI_i)$ and $\mcL(\mcO_i)$ are the sets of linear operators on finite-dimensional Hilbert spaces $\mcI_i$ and $\mcO_i$, respectively, and $A$ is an index set.}
It is implemented by a quantum circuit using auxiliary Hilbert spaces $\mcA_i$ for $i\in [n]$, a quantum state $\phi\in\mcL(\mcI_1\otimes \mcA_1)$, quantum channels $\Phi_i: \mcL(\mcO_i \otimes \mcA_i) \to \mcL(\mcI_{i+1}\otimes \mcA_{i+1})$ for $i\in [n-1]$ and a POVM measurement $\{M_a\}_{a\in A} \subset \mcL(\mcO_n\otimes \mcA_n)$ such that
\begin{align}
    p_a
    &= \mcT_a(\Lambda_1,\ldots,\Lambda_n)\\
    &\coloneqq \mathrm{Tr}[M_a (\Lambda_n\otimes \1_{\mcL(\mcA_n)}) \circ \Phi_{n-1}\circ \nonumber\\
    &\hspace{30pt}\cdots \circ (\Lambda_2\otimes \1_{\mcL(\mcA_2)}) \circ \Phi_1\circ (\Lambda_1\otimes \1_{\mcL(\mcA_1)})(\phi)],
    \label{eq:def_superinstrument}
\end{align}
where $\{\mcT_a\}_{a\in A}$ represents a quantum tester.
The probability distribution to obtain an estimator $\hat{U}_a$ using $n$ queries to $\mcU$ is given by
\begin{align}
    p(\hat{U}_a \mid U) = \mcT_a(\mcU^{\times n}).
\end{align}
See Fig.~\ref{fig:unitary_estimation} for the overview of unitary estimation.

The standard figure of merit for unitary estimation is the \emph{average-case fidelity}, defined by
\begin{align}
    F_\mathrm{est}(\{\mcT_a\}_{a\in A}) \coloneqq \int_{\SU(d)} \dd U \sum_{a\in A} p(\hat{U}_a \mid U) F_\mathrm{ch}(U, \hat{U}_a),
\end{align}
where $\dd U$ is the Haar measure on $\SU(d)$ and $F_\mathrm{ch}(U, \hat{U}_a)$ is a channel fidelity~\cite{raginsky2001fidelity} between two unitary channels corresponding to $U$ and $\hat{U}_a$, defined by
\begin{align}
    F_\mathrm{ch}(U, \hat{U}_a)\coloneqq {1 \over d^2} \abs{\Tr(U^\dagger \hat{U}_a)}^2.
\end{align}
We consider the maximium value of the average-case fidelity for a given number of $n$ and $d$ denoted by $F_\mathrm{est}(n,d)$, i.e.,
\begin{align}
    F_\mathrm{est}(n,d)\coloneqq \max_{\{\mcT_a\}_{a\in A}:\text{ quantum tester}} F_\mathrm{est}(\{\mcT_a\}_{a\in A}).
\end{align}

\subsection{Parallel covariant protocol}

References~\cite{chiribella2005optimal, bisio2010optimal} have shown that the maximum average-case fidelity of unitary estimation is attained by a parallel covariant protocol $\{\mcT_{\hat{U}}\}_{\hat{U}\in\SU(d)}$ whose index set is given by $\SU(d)$.
The protocol is given by
\begin{align}
    \mcT_{\hat{U}}(\mcU^{\times n}) \coloneqq \Tr[M_{\hat{U}} (\mcU^{\otimes n}\otimes \1_{\mcL(\mcA)})(\ketbra{\phi_{\vec{v}}})] \dd \hat{U},
\end{align}
where $\ket{\phi_{\vec{v}}}$ is a quantum state defined by
\begin{align}
    \ket{\phi_{\vec{v}}} = U_\mathrm{Sch}^{\dagger} \bigoplus_{\mu\in\young{d}{n}} {v_\mu \over \sqrt{\dim \mcH_\mu}} \ket{W_\mu},
\end{align}
$\{M_{\hat{U}} \dd \hat{U}\}_{\hat{U}\in\SU(d)}$ is a POVM measurement with continuous measurement outcome defined by
\begin{align}
    M_{\hat{U}} &\coloneqq \ketbra{\eta_{\hat{U}}},\\
    \ket{\eta_{\hat{U}}} &\coloneqq U_\mathrm{Sch}^{\dagger} \bigoplus_{\mu\in\young{d}{n}} [\rho_{\mu}(\hat{U}) \otimes \1_{\mcH_\mu \CC^{m_\mu} \CC^{m_\mu}}] \ket{W_\mu},
\end{align}
$\vec{v}\coloneqq (v_\mu)_{\mu\in\young{d}{n}}$ is a $\abs{\young{d}{n}}$-dimensional complex vector satisfying
\begin{align}
    \abs{\vec{v}}^2\coloneqq \sum_{\mu\in\young{d}{n}}\abs{v_\mu}^2=1,
\end{align}
and $\ket{W_\mu}\in \mcH_\mu \otimes \mcH_\mu \otimes \CC^{m_\mu} \otimes \CC^{m_\mu}$ is an unnormalized vector defined by
\begin{align}
    \ket{W_\mu} \coloneqq \sum_{q_\mu} \ket{q_\mu}_{\mcH_\mu} \otimes \ket{q_\mu}_{\mcH_\mu} \otimes \ket{\mathrm{arb}}_{\CC^{m_\mu}\CC^{m_\mu}}
\end{align}
using an arbitrary normalized vector $\ket{\mathrm{arb}}\in \CC^{m_\mu}\otimes \CC^{m_\mu}$.

The average-case channel fidelity of the parallel covariant protocol is given by~\cite{chiribella2005optimal, yang2020optimal, yoshida2024one}
\begin{align}
    \vec{v}^{\dagger} M_\mathrm{est} \vec{v},
\end{align}
where $M_\mathrm{est}$ is a $\abs{\young{d}{n}}\times \abs{\young{d}{n}}$ matrix defined by
\begin{align}
    (M_\mathrm{est})_{\mu\nu} \coloneqq {\#[(\mu+_d\square)\cap (\nu+_d\square)] \over d^2}
\end{align}
for all $\mu, \nu\in\young{d}{n}$, $\mu+_d\square$ is defined by
\begin{align}
    \mu+_d\square \coloneqq \{\mu+e_i \mid i\in[d]\}\cap \young{d}{n+1}
\end{align}
and $e_i$ is a tuple of $d$ numbers defined by $(e_i)_j\coloneqq \delta_{i=j}$.
Therefore, the maximum fidelity of unitary estimation is given by
\begin{align}
    F_\mathrm{est}(n,d)
    &= \max_{\abs{\vec{v}}^2=1} \vec{v}^\dagger M_\mathrm{est} \vec{v}\\
    &= \max \eig M_{\mathrm{est}},\label{eq:estimation_fidelity_maximal_eigenvalue}
\end{align}
where $\max \eig M_{\mathrm{est}}$ represents the maximum eigenvalue of $M_\mathrm{est}$.
Note that $M_\mathrm{est}$ can also be represented as
\begin{align}
    (M_\mathrm{est})_{\mu\nu}
    =
    \begin{cases}
        {\#(\mu+_d\square) \over d^2} & (\mu=\nu)\\
        {1\over d^2} & (\exists i\neq j \text{ s.t. } \mu=\nu+f_{ij})\\
        0 & (\text{otherwise})
    \end{cases},
\end{align}
where $f_{ij}$ is defined by $f_{ij}\coloneqq e_i-e_j$.

\subsection{One-to-one correspondence between deterministic port-based teleportation and unitary estimation}

\begin{figure}
    \centering
    \includegraphics[width=\linewidth]{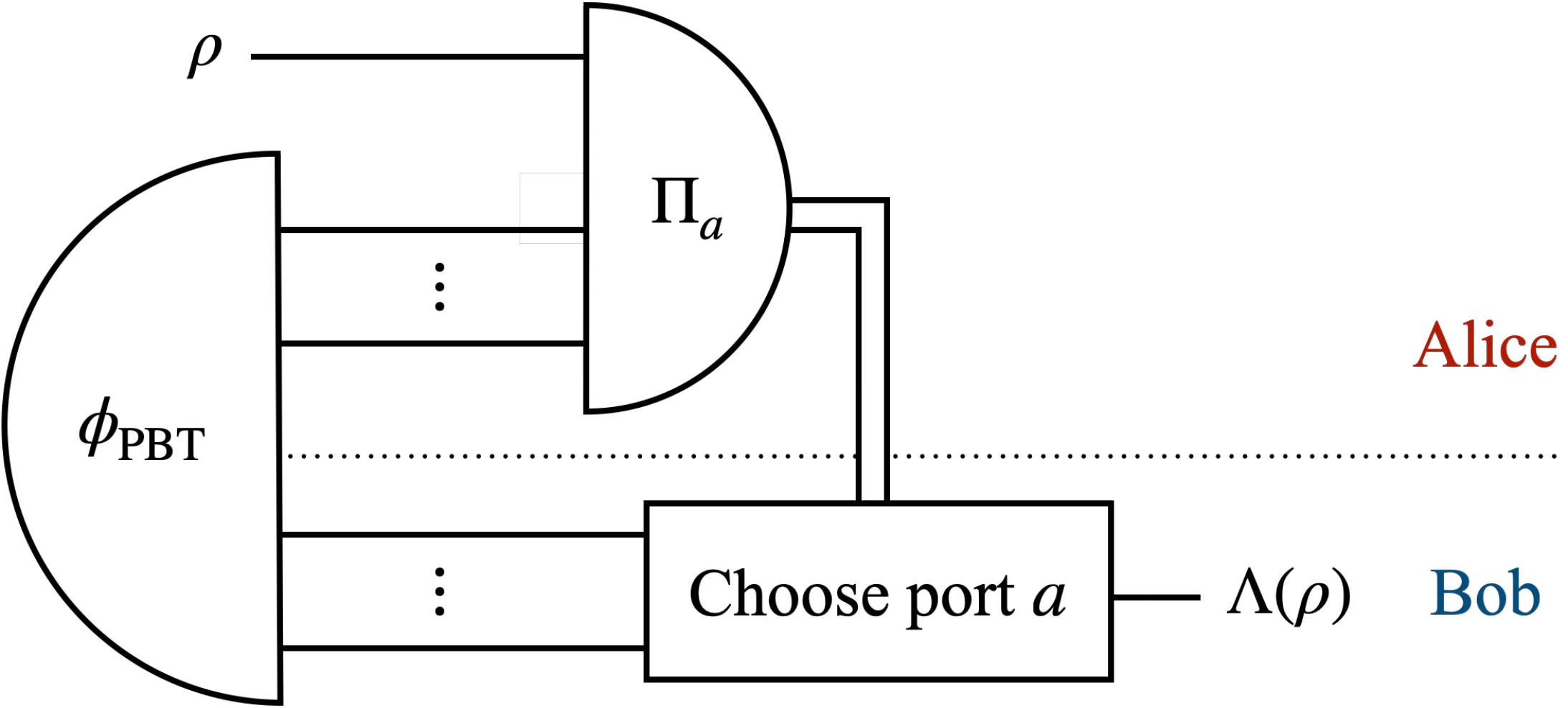}
    \caption{The quantum circuit expressing the deterministic port-based teleportation (dPBT).
    ALice and Bob shares an entangled state $\phi_\mathrm{PBT}$, and Alice applies a joint measurement $\{\Pi_a\}_a$ on her input state $\rho$ and her share of $\phi_\mathrm{PBT}$.
    Alice sends the measurement outcome $a$ to Bob, and Bob chooses port $a$ on his share of $\phi_\mathrm{PBT}$ to obtain a quantum state $\Lambda(\rho)$ [see also Eq.~\eqref{eq:PBT}].}
    \label{fig:PBT}
\end{figure}

The optimal fidelity of unitary estimation has a close connection with another task called deterministic port-based teleportation (dPBT)~\cite{yoshida2024one}.
The task of dPBT is summarized as follows (see also Fig.~\ref{fig:PBT}).
The sender (Alice) wants to teleport an arbitrary qudit state $\rho\in \mcL(\mcA_0)$ to the receiver (Bob) using a shared $2N$-qudit entangled state $\phi_\mathrm{PBT}\in \mcL(\mcA^N \otimes \mcB^N)$, where $\mcA^N$ and $\mcB^N$ are the joint Hilbert spaces defined by $\mcA^N\coloneqq \bigotimes_{i=1}^{N} \mcA_i$ and $\mcB^N\coloneqq \bigotimes_{i=1}^{N} \mcB_i$ for $\mcA_0 \simeq \mcA_i\simeq \mcB_i\simeq (\CC^d)^{\otimes N}$.
The Hilbert spaces $\mcA_i$ and $\mcB_i$ for $i\in [n]$ are called the ports.
Alice measures a quantum state in $\mcA_0\otimes \mcA^N$ with a POVM measurement $\{\Pi_a\}_{a=1}^{N}$, and sends the measurement outcome $a$ to Bob.
Bob chooses the port $a$ based on the measurement outcome $a$ to obtain the quantum state $\Lambda(\rho)$, where $\Lambda: \mcL(\mcA_0) \to \mcL(\mcB_0)$ is the teleportation channel defined by
\begin{align}
\label{eq:PBT}
    \Lambda(\rho)\coloneqq \sum_{a=1}^{N} \Tr_{\mcA_0\mcA^N\overline{\mcB_a}}[(\Pi_a\otimes \1_{\mcB^N})(\rho\otimes \phi_\mathrm{PBT})],
\end{align}
where $\overline{\mcB_a}\coloneqq \bigotimes_{i\neq a} \mcB_i$.

The performance of dPBT is evaluated by the teleportation fidelity~\cite{horodecki1999general}, which is the channel fidelity~\cite{raginsky2001fidelity} between $\Lambda$ and the identity channel $\1_{\mcL(\CC^d)}$ given by
\begin{align}
    F_\mathrm{PBT}(\phi_\mathrm{PBT}, \{\Pi_a\}_{a=1}^{N})\coloneqq {1\over d^2} \sum_i \abs{\Tr(K_i)}^2,
\end{align}
where $\{K_i\}_i$ is the set of Kraus operators~\cite{nielsen2010quantum} of $\Lambda$ satisfying $\Lambda(\rho) = \sum_i K_i \rho K_i^\dagger$.
We consider the maximum value of the teleportation fidelity for a given $N$ and $d$ denoted by $F_\mathrm{PBT}(N,d)$, i.e.,
\begin{align}
    F_\mathrm{PBT}(N,d) \coloneqq \max_{\substack{\phi_\mathrm{PBT}: \text{ quantum state}\\\{\Pi_a\}_{a=1}^{N}:\text{ POVM}}} F_\mathrm{PBT}(\phi_\mathrm{PBT}, \{\Pi_a\}_{a=1}^{N}).
\end{align}
Reference~\cite{yoshida2024one} shows that optimal teleportation fidelity of dPBT with $N=n+1$ is equivalent to that of the estimation fidelity of unitary estimation with $n$ queries:
\begin{align}
    F_\mathrm{PBT}(N=n+1,d) = F_\mathrm{est}(n,d),
\end{align}
with an explicit construction of the optimal protocol of dPBT (unitary estimation) from a given optimal protocol of unitary estimation (dPBT).

\section{Optimal performance of unitary estimation}
\label{sec:main_result}
\subsection{Asymptotically optimal unitary estimation in $\SU(3)$}

We first show the following Theorem on the asymptotic optimal fidelity of unitary estimation.

\begin{theorem}
\label{thm:asymptotically_optimal_unitary_estimation}
    For $d=2,3$, the asymptotic optimal fidelity of $d$-dimensional unitary estimation is given by
    \begin{align}
        F_\mathrm{est}(n,d) = 1-{\lambda_\mathrm{min}(\Omega_{d-1}) \over d n^2}+O(n^{-3}),
    \end{align}
    where $\lambda_\mathrm{min}(\Omega_{d-1})$ is the minimum eigenvalue of the following Dirichlet Laplacian problem:
    \begin{align}
    \label{eq:dirichlet_laplacian}
    \begin{split}
        \Delta u &= -\lambda u,\\
        u\vert_{\partial \Omega_{d-1}} &= 0,\\
        u&\not\equiv 0,
    \end{split}
    \end{align}
    $\Delta$ is the $(d-1)$-dimensional Laplacian operator and $\Omega_{d-1}$ is a $(d-1)$-polytope defined by
    \begin{align}
    \label{eq:def_Omega}
        \Omega_{d-1}\coloneqq \Bigg\{\vec{x}\in \RR^d \;\Bigg|\; x_1\geq \cdots \geq x_d\geq 0, \;\sum_{i=1}^{d} x_i = 1\Bigg\},
    \end{align}
    and $\partial \Omega_{d-1}$ is the boundary of $\Omega_{d-1}$.
\end{theorem}

This theorem reproduces the asymptotic rate of the optimal unitary estimation for $d=2$ given by~\cite{bagan2004entanglement}
\begin{align}
    F_\mathrm{est}(n,d=2) = 1-{\pi^2\over n^2}+O(n^{-3}),
\end{align}
since $\Omega_1$ is a line segment of length ${1\over \sqrt{2}}$ and the first Dirichlet eigenvalue is given by $\lambda_\mathrm{min}(\Omega_1) = 2 \pi^2$.
For $d=3$, we show the following analytical formula for the asymptotic fidelity of unitary estimation in $\SU(3)$:
\begin{corollary}
\label{cor:h(d)_d=3}
    \begin{align}
        F_\mathrm{est}(n, d=3) = 1-{56 \pi^2 \over 9 n^2}+O(n^{-3}).
    \end{align}
\end{corollary}

\begin{figure}
    \centering
    \includegraphics[width=0.5\linewidth]{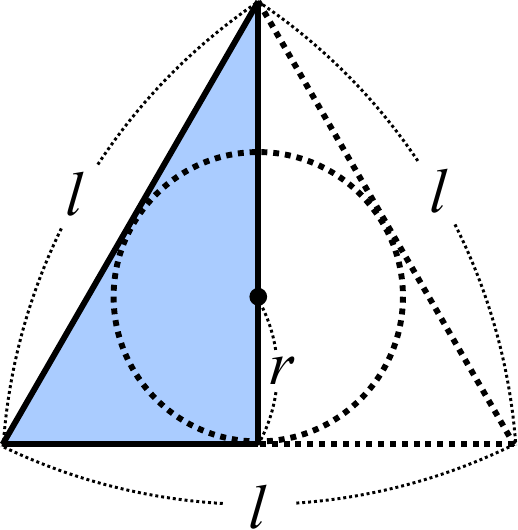}
    \caption{The region $\Omega_{2}$ defined in \eqref{eq:def_Omega} is the hemi-equilateral triangle given by the half of the equilateral triangle with the side length of $l={\sqrt{6}\over 3}$ as shown in the blue region in this figure.
    The inradius $r$ of the equilateral triangle is also shown in the figure.}
    \label{fig:triangle}
\end{figure}

\begin{proof}
    The region $\Omega_{2}$ is the hemi-equilateral triangle given by the half of the equilateral triangle with the side length of $l = {\sqrt{6} \over 3}$ (see Fig.~\ref{fig:triangle}).
    The minimum eigenvalue of the hemi-equilateral triangle given by the half of the equilateral triangle with the inradius $r$ is given by~\cite{lame1833memoire,mccartin2003eigenstructure}
    \begin{align}
        {28\pi^2 \over 27 r^2}.
    \end{align}
    Substituting $r={l\over 2\sqrt{3}} = {1\over 3\sqrt{2}}$, the minimum eigenvalue of $\Omega_2$ is given by
    \begin{align}
        \lambda_\mathrm{min}(\Omega_2) = {56\pi^2 \over 3}.
    \end{align}
    Therefore, we obtain
    \begin{align}
        F_\mathrm{est}(n, d=3)
        &= 1-{\lambda_\mathrm{min}(\Omega_2) \over 3 n^2}+O(n^{-3})\\
        &=1-{56 \pi^2 \over 9 n^2}+O(n^{-3}).
    \end{align}
\end{proof}

\subsection{Proof of Theorem~\ref{thm:asymptotically_optimal_unitary_estimation}}
\label{appendix_sec:proof_asymptotically_optimal_unitary_estimation}

The outline of the proof of Theorem~\ref{thm:asymptotically_optimal_unitary_estimation} is described as follows.
Since Eq.~\eqref{eq:estimation_fidelity_maximal_eigenvalue} holds, defining $h_{n,d}$ by
\begin{align}
    h_{n,d}\coloneqq \min \eig [\1-M_{\mathrm{est}}(n,d)] n^2,
\end{align}
we obtain
\begin{align}
    F_\mathrm{est}(n,d) = 1-{h_{n,d} \over n^2}.
\end{align}
Thus, Theorem~\ref{thm:asymptotically_optimal_unitary_estimation} is equivalent to
\begin{align}
    h_{n,d} = {\lambda_\mathrm{min}(\Omega_{d-1}) \over d} + O(n^{-1})
\end{align}
for $d=2,3$.
First, using a similar argument as shown in Ref.~\cite{christandl2021asymptotic}, we construct a unitary estimation protocol achieving the fidelity
\begin{align}
    F_\mathrm{est}(\{\mcT_a\}_a) = 1-{\lambda_\mathrm{min}(\Omega_{d-1}) \over dn^2} + O(n^{-3})
\end{align}
using the first eigenfunction $u$ of the Dirichlet Laplacian shown in Eq.~\eqref{eq:dirichlet_laplacian}.
This shows an upper bound on $h_{n,d}$ given by
\begin{align}
\label{eq:h(n,d)_upper_bound}
    h_{n,d} \leq {\lambda_\mathrm{min}(\Omega_{d-1}) \over d}+ O(n^{-1}).
\end{align}

Next, we use the finite element method to show the converse bound on $h_{n,d}$.
We first show a lower bound on $h_{n,d}$ by the minimum eigenvalue of a discrete Dirichlet Laplacian problem.
We consider an undirected graph with boundary $\overline{G}_{n,d} = (\overline{V}_{n,d}, \overline{E}_{n,d})$, where $\overline{V}_{n,d}$ is the set of vertices that is given by a disjoint union of the interior set $V_{n,d}$ and the boundary set $\partial V_{n,d}$:
\begin{align}
    \overline{V}_{n,d} = V_{n,d} \sqcup \partial V_{n,d},
\end{align}
and $\overline{E}_{n,d}$ is the set of edges.
The Laplacian $L(\overline{G}_{n,d})$ of $\overline{G}_{n,d}$ is defined by a $\abs{\overline{V}_{n,d}} \times \abs{\overline{V}_{n,d}}$ matrix given by
\begin{align}
    [L(\overline{G}_{n,d})]_{pq}\coloneqq
    \begin{cases}
        \deg p & (p=q)\\
        -\delta_{\{p,q\}\in \overline{E}_{n,d}} & (p\neq q)
    \end{cases}
\end{align}
for all $p,q\in \overline{V}_{n,d}$, where $\deg p$ is defined by
\begin{align}
    \deg p \coloneqq \#\{q\in \overline{V}_{n,d} \mid \{p,q\}\in \overline{E}_{n,d}\}.
\end{align}
We define the minimum eigenvalue $\lambda_\mathrm{min}(\overline{G}_{n,d})$ of the following discrete Dirichlet Laplacian problem given by
\begin{align}
    L(\overline{G}_{n,d}) \vec{u} &= \lambda \vec{u},\\
    u_{p} &= 0 \quad \forall {p}\in \partial \overline{V}_{n,d},
\end{align}
where $\vec{u} = (u_{p})_{p\in \overline{V}_{n,d}}$ is a $\abs{\overline{V}_{n,d}}$-dimensional real vector.
For a proper choice of the graph $\overline{G}_{n,d}$, we show a lower bound of $h_{n,d}$ given by
\begin{align}
    \lambda_\mathrm{min}(\overline{G}_{n,d})
    &\leq \min \eig [\1-M_\mathrm{est}(n,d)]d^2\\
    &= {d^2\over n^2} h_{n,d}.
\end{align}
We define a region $\Omega_{d-1}^{(n)}$ approximating the region $\Omega_{d-1}$ as
\begin{align}
\label{eq:inner_approximation_Omega}
    \Omega_{d-1}^{(n)} \subset r_n \Omega_{d-1}
\end{align}
for a proper choice of $r_n = 1+O(n^{-1})$, where $r_n \Omega_{d-1}$ is obtained by a scaling transformation by a factor of $r_n$ centered at the centroid $\vec{g}$, i.e.,
\begin{align}
    r_n \Omega_{d-1} &\coloneqq \{r_n (\vec{x}-\vec{g})+\vec{g} \mid x\in \Omega_{d-1}\},\\
    \vec{g}&\coloneqq {1\over \mathrm{Vol}(\Omega_{d-1})}\int_{\Omega_{d-1}} \vec{x} \dd \vec{x},
\end{align}
and $\mathrm{Vol}(\Omega_{d-1})$ is the volume of $\Omega_{d-1}$.
Using the argument of the finite element method, we show an upper bound on $\lambda_\mathrm{min}(\Omega_{d-1}^{(n)})$ given by
\begin{align}
\label{eq:lambda_FEM}
    \lambda_\mathrm{min}(\Omega_{d-1}^{(n)}) \leq {n^2\over d} \lambda_\mathrm{min}(\overline{G}_{n,d}) + O(n^{-2}).
\end{align}
From Eq.~\eqref{eq:inner_approximation_Omega}, $\lambda_\mathrm{min}(\Omega_{d-1}^{(n)})$ satisfies
\begin{align}
    \lambda_\mathrm{min}(\Omega_{d-1}^{(n)})
    &\geq {\lambda_\mathrm{min}(\Omega_{d-1}) \over r_n^2}\\
    &= \lambda_\mathrm{min}(\Omega_{d-1}) + O(n^{-1}).
\end{align}
Therefore, we obtain
\begin{align}
\label{eq:h(n,d)_lower_bound}
    h_{n,d}
    &\geq {\lambda_{\mathrm{min}}(\Omega_{d-1}) \over d}+O(n^{-1}),
\end{align}
which completes the proof of Theorem~\ref{thm:asymptotically_optimal_unitary_estimation}.

\subsubsection{Upper bound on $h_{n,d}$}
\label{subsubsec:proof_upper_bound}

We show the upper bound~\eqref{eq:h(n,d)_upper_bound} on $h_{n,d}$ for an arbitrary $d$.
Suppose $u$ is the minimum eigenfunction of the Dirichlet Laplacian problem~\eqref{eq:dirichlet_laplacian}.
Then, the minimum eigenvalue $\lambda_\mathrm{min}(\Omega)$ is given by the Rayleigh quotient
$R_{\Omega_{d-1}}(u)$ defined by
\begin{align}
\label{eq:rayleigh}
    R_{\Omega_{d-1}}(u)\coloneqq {\int_{\Omega_{d-1}} \abs{\nabla u}^2 \over \int_{\Omega_{d-1}} u^2}.
\end{align}
Defining $v_\mu$ by
\begin{align}
    v_\mu\coloneqq {u(\mu/n) \over \sqrt{\sum_{\mu\in \young{d}{n}} u(\mu/n)^2}},
\end{align}
we show that
\begin{align}
    \vec{v}^\dagger M_\mathrm{est} \vec{v} = 1-{R_{\Omega_{d-1}}(u)\over dn^2} +O(n^{-3})
\end{align}
to complete the proof of Eq.~\eqref{eq:h(n,d)_upper_bound}.
To this end, we first evaluate $(M_{\mathrm{est}})_{\mu \nu}$ for $\mu, \nu \in \young{d}{n} \setminus n\partial \Omega_{d-1}$, which is given by
\begin{align}
    (M_{\mathrm{est}})_{\mu \nu}
    =
    \begin{cases}
        {1\over d} & (\mu=\nu)\\
        {1\over d^2} & (\exists i\neq j \text{ s.t. } \mu = \nu+f_{ij})\\
        0 & (\text{otherwise})
    \end{cases}.
\end{align}
Therefore, we obtain
\begin{align}
    &\vec{v}^\dagger M_\mathrm{est} \vec{v}\nonumber\\
    &= {{1\over d^2} \sum_{\mu\in \young{d}{n}} \left[d u(\mu/n)^2 + \sum_{i\neq j} u(\mu/n)u(\mu/n+f_{ij}/n)\right] \over \sum_{\mu\in \young{d}{n}} u(\mu/n)^2}\\
    &= {{1\over d^2} \sum_{\mu\in \young{d}{n}} \sum_{i,j=1}^{d}  u(\mu/n)u(\mu/n+f_{ij}/n) \over \sum_{\mu\in \young{d}{n}} u(\mu/n)^2}.
\end{align}
As shown in Theorem~1.4 and Lemma~6.2 of Ref.~\cite{christandl2021asymptotic}, this is evaluated as
\begin{align}
    &\vec{v}^\dagger M_\mathrm{est} \vec{v}\nonumber\\
    &= 1-{\int_{\Omega_{d-1}} u(\vec{x}) \Delta u(\vec{x})  \dd \vec{x}\over dn^2\int_{\Omega_{d-1}} u(\vec{x})^2 \dd \vec{x}} + O(n^{-3})\\
    &= 1-{R_{\Omega_{d-1}}(u)\over dn^2} +O(n^{-3}).
\end{align}

\subsubsection{Lower bound on $h_{n,d}$}
\label{subsubsec:proof_lower_bound}

We show the lower bound~\eqref{eq:h(n,d)_lower_bound} on $h_{n,d}$ for $d=2,3$.
We define an undirected graph $G_{n,d} = (V_{n,d}, E_{n,d})$ with the set of vertices $V_{n,d}$ and the set of edges given by
\begin{align}
    V_{n,d} &\coloneqq \left\{\mu/n \;\middle| \; \mu \in \young{d}{n}\right\} \subset \RR^d,\\
    E_{n,d} &\coloneqq \left\{\{p,q\} \;\middle|\; p,q\in V_{n,d}, \exists i\neq j \; \mathrm{s.t.} \; q = p + f_{ij} /n\right\}.
\end{align}
We also define an undirected graph $\overline{G}_{n,d} = (\overline{V}_{n,d}, \overline{E}_{n,d})$ by (see also Fig.~\ref{fig:grid_with_caption})
\begin{align}
    \overline{V}_{n,d} &\coloneqq V_{n,d} \sqcup \partial \overline{V}_{n,d}\\
    \partial \overline{V}_{n,d}&\coloneqq \left\{v+f_{ij}/n \middle| v\in V_{n,d}, \exists i\neq j\right\}\setminus V_{n,d},\\
    \overline{E}_{n,d} &\coloneqq \left\{\{v,w\} \;\middle|\; v, w\in \overline{V}_{n,d}, \exists i\neq j \; \mathrm{s.t.} \; w = v + f_{ij}/n\right\}.
\end{align}
Defining a $\abs{V_{n,d}}\times \abs{V_{n,d}}$ matrix $L_{n,d}$ by
\begin{align}
    [L_{n,d}]_{\mu\nu} \coloneqq [L(\overline{G}_{n,d})]_{p_\mu p_\nu}
\end{align}
by setting $p_\mu = \mu/n \in V_{n,d}$ and $p_\nu = \nu/n \in V_{n,d}$ for $\mu, \nu\in\young{d}{n}$, $\lambda_\mathrm{min}(\overline{G}_{n,d})$ is given by
\begin{align}
    \lambda_\mathrm{min}(\overline{G}_{n,d}) = \min \eig L_{n,d}.
\end{align}
Then, we can show that
\begin{align}
     L_{n,d} \leq [\1-M_{\mathrm{est}}(n,d)] d^2
\end{align}
since for all $\mu\neq \nu \in V$,
\begin{align}
    [L_{n,d}]_{\mu\nu}
    &= -\#[(\mu+_d\square)\cap (\nu+_d\square)]\\
    &= [\1-M_{\mathrm{est}}(n,d)]_{\mu\nu} d^2,
\end{align}
and for all $\mu\in V$,
\begin{align}
    &[L_{n,d}]_{\mu\mu}\nonumber\\
    &= \sum_{\mu'\neq \mu} \#[(\mu+_d\square) \cap (\mu'+_d\square)]\\
    &= \left[\sum_{\mu'\in V}\#[(\mu+_d\square) \cap (\mu'+_d\square)] - \#(\mu+_d\square)\right]\\
    &= \left[\sum_{i,j=1}^{d} \delta_{\mu+e_i\in \young{d}{n+1}} \delta_{\mu+f_{ij}\in \young{d}{n}} - \#(\mu+\square)\right]\\
    &\leq \left[d^2 - \#(\mu+\square)\right]\\
    &= [\1-M_{\mathrm{est}}(n,d)]_{\mu\mu} d^2
\end{align}
holds.
Therefore, we obtain
\begin{align}
    h_{n,d} \geq {n^2\over d^2} \lambda_\mathrm{min}(\overline{G}_{n,d}).
\end{align}

We define the set of $(d-1)$-dimensional simplicies defined by
\begin{align}
    \Xi_{n,d} \coloneqq \left\{\overline{p_1\cdots p_d} \; \middle| \; \forall i\neq j, \{p_i,p_j\}\in \overline{E}_{n,d}\right\},
\end{align}
where $\overline{p_1\cdots p_d}$ is the $(d-1)$-dimensional simplex defined by
\begin{align}
    \overline{p_1\cdots p_d}\coloneqq \left\{\sum_{i=1}^{d} a_i p_i \; \middle| \; a_i\geq 0, \sum_{i=1}^{d} a_i = 1\right\}.
\end{align}
For $d=2,3$, defining a region $\Omega^{(n)}_{d-1} \subset \RR^{d-1}$ by
\begin{align}
    \Omega^{(n)}_{d-1}\coloneqq \bigcup \Xi_{n,d},
\end{align}
the region $\Omega^{(n)}_{d-1}$ approximates the region $\Omega_{d-1}$ as Eq.~\eqref{eq:inner_approximation_Omega}, and $\Xi_{n,d}$ is a triangulation of $\Omega_{d-1}^{(n)}$.
Therefore, by using the finite element method, the minimum eigenvalue $\lambda_\mathrm{min}(\Omega_{1}^{(n)})$ of the Dirichlet Laplacian is upper bounded by the minimum eigenvalue of $M(\Xi_{n,d})^{-1} K(\Xi_{n,d})$, where $K(\Xi_{n,d})$ and $M(\Xi_{n,d})$ are the stiffness matrix and the mass matrix defined in Eqs.~\eqref{eq:def_stiffness} and \eqref{eq:def_mass} (see Appendix~\ref{appendix_sec:review_FEM} for the detail).
We show Eq.~\eqref{eq:lambda_FEM} for the cases of $d=2,3$ below.

For $d=2$, since all the segments in $\Xi_{n,2}$ have length ${\sqrt{2}\over n}$, the stiffness matrix and the mass matrix are given by
\begin{align}
    [K(\Xi_{n,2})]_{pq} &= 
    \begin{cases}
        \sqrt{2}n & (p=q)\\
        -{n\over \sqrt{2}} & (\{p,q\}\in E_{n,2})\\
        0 & (\mathrm{otherwise})
    \end{cases},\\
    [M(\Xi_{n,2})]_{pq} &= 
    \begin{cases}
        {2\sqrt{2} \over 3n}& (p=q)\\
        {\sqrt{2} \over 6n} & (\{p,q\}\in E_{n,2})\\
        0 & (\mathrm{otherwise})
    \end{cases}
\end{align}
for all $v, w\in E'_{n,2}$.
Since $L_{n,2}$ is given by
\begin{align}
    [L_{n,2}]_{pq} = \begin{cases}
        2 & (p=q)\\
        -1 & (\{p,q\}\in E_{n,2})\\
        0 & (\mathrm{otherwise})
    \end{cases},
\end{align}
we can write
\begin{align}
    K(\Xi_{n,2}) &= {n\over \sqrt{2}} L_{n,2},\\
    M(\Xi_{n,2}) &= {\sqrt{2}\over n} \left[\1 - {L_{n,2}\over 6}\right].
\end{align}
From Theorem~\ref{thm:FEM}, we obtain
\begin{align}
    \lambda_\mathrm{min}(\Omega_1^{(n)})
    &\leq \min \eig M(\Xi_{n,2})^{-1} K(\Xi_{n,2})\\
    &= {n^2 \over 2} \min \eig \left[\1- {L_{n,2}\over 6}\right]^{-1}L_{n,2}\\
    &= {n^2\over 2}{\lambda_\mathrm{min}(\overline{G}_{n,2}) \over 1-{1\over 6} \lambda_\mathrm{min}(\overline{G}_{n,2})}.
\end{align}
Since $\lambda_\mathrm{min}(\overline{G}_{n,2}) = O(n^{-2})$ holds, we obtain
\begin{align}
    \lambda_\mathrm{min}(\Omega_1^{(n)})\leq {n^2\over 2}\lambda_\mathrm{min}(\overline{G}_{n,2}) + O(n^{-2}),
\end{align}
i.e., Eq.~\eqref{eq:lambda_FEM} holds.

\begin{figure}
    \centering
    \includegraphics[width=\linewidth]{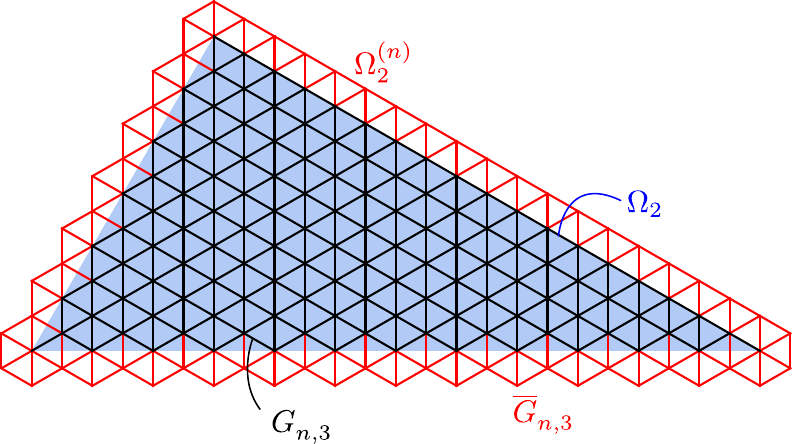}
    \caption{Illustration of the graphs $G'_{n,3}$, $\overline{G}'_{n,3}$ and the regions $\Omega_2$, $\Omega_2^{(n)}$ for $n=36$.}
    \label{fig:grid_with_caption}
\end{figure}
For $d=3$, since all the triangles in $\Xi_{n,3}$ are equilateral triangles with the side length ${\sqrt{2} \over n}$, the stiffness matrix and the mass matrix are given by~\cite{jumonji2008eigenvalue}
\begin{align}
    [K(\Xi_{n,3})]_{pq} &= 
    \begin{cases}
        2\sqrt{3} & (p=q)\\
        -{1\over \sqrt{3}} & (\{p,q\}\in E_{n,3})\\
        0 & (\mathrm{otherwise})
    \end{cases},\\
    [M(\Xi_{n,3})]_{pq} &= 
    \begin{cases}
        {\sqrt{3} \over 2 n^2} & (p=q)\\
        {\sqrt{3} \over 12n^2} & (\{p,q\}\in E_{n,3})\\
        0 & (\mathrm{otherwise})
    \end{cases}
\end{align}
for all $v, w\in E_{n,3}$.
Since $L_{n,3}$ is given by
\begin{align}
    [L_{n,3}]_{pq} = \begin{cases}
        6 & (p=q)\\
        -1 & (\{p,q\}\in E_{n,3})\\
        0 & (\mathrm{otherwise})
    \end{cases},
\end{align}
we can write
\begin{align}
\label{eq:stiffness_d=3}
    K(\Xi_{n,3}) &= {1\over \sqrt{3}} L_{n,3},\\
\label{eq:mass_d=3}
    M(\Xi_{n,3}) &= {\sqrt{3} \over n^2} \left[\1- {L_{n,3}\over 12}\right].
\end{align}
From Theorem~\ref{thm:FEM}, we obtain
\begin{align}
    \lambda_\mathrm{min}(\Omega_2^{(n)})
    &\leq \min \eig M(\Xi_{n,3})^{-1} K(\Xi_{n,3})\\
    &= {n^2 \over 3} \min \eig \left[\1- {L_{n,3}\over 12}\right]^{-1}L_{n,3}\\
    &= {n^2\over 3}{\lambda_\mathrm{min}(\overline{G}_{n,3}) \over 1-{1\over 12}\lambda_\mathrm{min}(\overline{G}_{n,3})}.
\end{align}
Since $\lambda_\mathrm{min}(\overline{G}_{n,3}) = O(n^{-2})$ holds, we obtain
\begin{align}
    \lambda_\mathrm{min}(\Omega_2^{(n)})\leq {n^2\over 3}\lambda_\mathrm{min}(\overline{G}_{n,3}) + O(n^{-2}),
\end{align}
i.e., Eq.~\eqref{eq:lambda_FEM} holds.
This completes the proof of Theorem~\ref{thm:asymptotically_optimal_unitary_estimation}.

\subsection{Near-optimal unitary estimation fidelity in arbitrary dimension based on Kahn's protocol}

Kahn~\cite{kahn2007fast} obtained an upper bound on $h(d)$ given by\footnote{There is a typo in Ref.~\cite{kahn2007fast} on the definition of $S_{d-1}$; see Appendix~\ref{appendix_subsec:kahn_derivation} for the derivation.}
\begin{align}
\label{eq:kahn_integral}
\begin{split}
    h(d) &\leq {2d(A_d -B_d) - (d+1) C_d \over d^2 D_d},\\
    A_d &\coloneqq \int_{S_{d-1}} \left(\sum_{i=1}^{d} (x_{\{i\}})^2 \right)\dd \vec{x},\\
    B_d&\coloneqq \int_{S_{d-1}} \left(\sum_{i=2}^{d} x_{\{i\}} x_{\{i-1\}}\right) \dd \vec{x},\\
    C_d&\coloneqq \int_{S_{d-1}} (x_{\{d\}})^2 \dd \vec{x},\\
    D_d&\coloneqq \int_{S_{d-1}} \prod_{i=1}^{d} x_i^2 \dd \vec{x},\\
    S_{d-1}&\coloneqq \{\vec{x}=(x_1,\ldots,x_d) \mid x_i\geq 0,  \sum_{i=1}^{d} i x_i = 1\},
\end{split}
\end{align}
where $x_{\{i\}}$ is defined by
\begin{align}
    x_{\{i\}}\coloneqq \prod_{j\neq i} x_j.
\end{align}
We rederive this formula in Appendix~\ref{appendix_subsec:kahn_derivation}.
We obtain analytical formulae for $A_d, B_d, C_d, D_d$ to show the following corollary (see Appendix~\ref{appendix_subsec:kahn_solution} for the detail).
\begin{corollary}
\label{cor:kahn}
    \begin{align}
        h(d)\leq {1\over 6} (d+1)(d-1)(3d-2)(3d-1).
    \end{align}
\end{corollary}
\begin{proof}
    This corollary directly follows from Lemma~\ref{lem:kahn_integral} in Appendix~\ref{appendix_subsec:kahn_solution}.
\end{proof}

\begin{figure*}
    \centering
    \includegraphics[width=\linewidth]{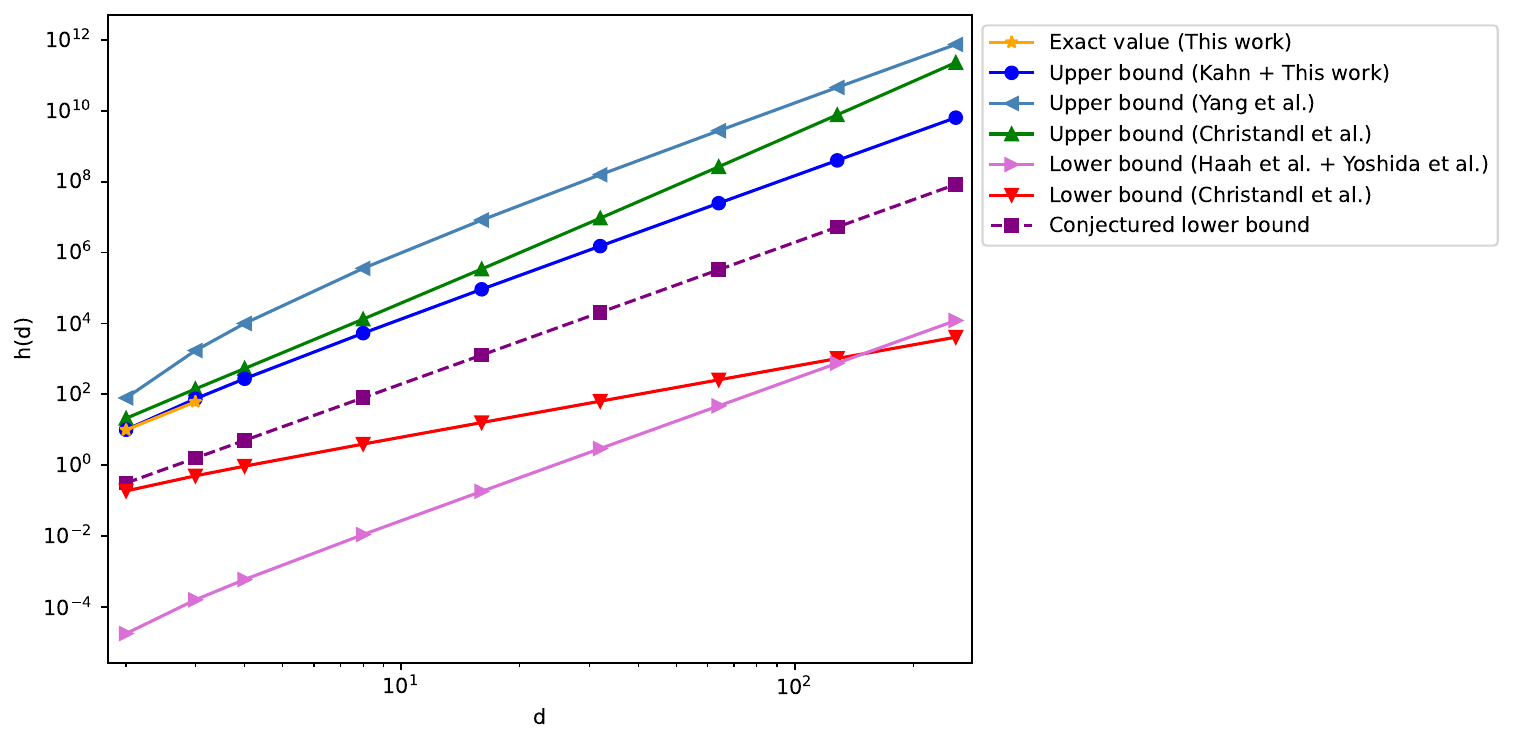}
    \caption{Comparison of the exact value of $h(d)$ for $d=2, 3$ ($h(3)$ is shown in Corollary~\ref{cor:h(d)_d=3}) and the upper bound of $h(d)$ (Corollary~\ref{cor:kahn}) with the previous works on upper and lower bounds on $h(d)$ shown in Eqs.~\eqref{eq:christandl}, \eqref{eq:yang} and \eqref{eq:haah}. The upper bound shown in Corollary~\ref{cor:kahn} (blue line) is a good estimate for the exact values of $h(d)$ for $d=2,3$, and is the best estimate for $d\geq 4$ compared to the previous works. The upper bound on Eq.~\eqref{eq:haah} is not plotted since the exact value of $\gamma$ is not given in Ref.~\cite{haah2023query}. We also plot the conjectured lower bound in Eq.~\eqref{eq:conjecture_lower_bound} as a dashed line.}
    \label{fig:comparison}
\end{figure*}

We compare Corollary~\ref{cor:h(d)_d=3} and \ref{cor:kahn} with the following previous works on upper and lower bounds on $h(d)$ in Fig.~\ref{fig:comparison}:

\begin{itemize}
    \item Christandl et al.~\cite{christandl2021asymptotic}:
    \begin{align}
    \label{eq:christandl}
        {d^2-1\over 16} \leq h(d) \leq {d^5 \over 4\sqrt{2}} + O(d^4).
    \end{align}
    \item Yang et al.~\cite{yang2020optimal}:
    \begin{align}
    \label{eq:yang}
        h(d)\leq 18 \pi^2 d^4 + O(d^3).
    \end{align}
    \item Haah et al.~\cite{haah2023query} + Yoshida et al.~\cite{yoshida2024one}:
    \begin{align}
    \label{eq:haah}
        \alpha(d^2 -\beta^2)^2 \leq h(d)\leq \gamma d^4
    \end{align}
    for constants
    \begin{align}
    \label{eq:def_alpha}
        \alpha &\coloneqq \left[{25-\ln 2 \over 20000({\pi/100}+\ln 2)}\right]^2 \approx 2.81\times 10^{-6},\\
    \label{eq:def_beta}
        \beta &\coloneqq \sqrt{50(\pi/100+\ln 2) \over 25-\ln 2} \approx 1.22,
    \end{align}
    and $\gamma>0$.
\end{itemize}
The upper bound and lower bound shown in Eq.~\eqref{eq:christandl} are better than Eqs.~\eqref{eq:yang} and \eqref{eq:haah} when $d$ is small ($d\lesssim 10^2$), while their scaling when $d\to \infty$ is not tight.
As shown in Fig.~\ref{fig:comparison}, the upper bound shown in Corollary~\ref{cor:kahn} is the best estimate compared to the previous works for all $d$.
For $d=2, 3$, upper bounds in Corollary~\ref{cor:kahn} are given by
\begin{align}
    h(2) &\leq 10,\\
    h(3) &\leq {224\over 3} \approx 74.67.
\end{align}
Since $h(2), h(3)$ are given by (Ref.~\cite{bagan2004entanglement} and Corollary~\ref{cor:h(d)_d=3})
\begin{align}
    h(2) &= \pi^2 \approx 9.87,\\
    h(3) &= {56\pi^2\over 9} \approx 61.41,
\end{align}
these upper bounds are good estimates of the true values.
The upper bound in Corollary~\ref{cor:kahn} is also asymptotically tight since $h(d)=\Theta(d^4)$ holds from Eqs.~\eqref{eq:yang} and \eqref{eq:haah}.

\section{Conclusion}
\label{sec:conclusion}
This work shows that the optimal fidelity of unitary estimation in $\SU(3)$ is given by $F_\mathrm{est}(n,d=3) = 1-{56\pi^2 \over 9n^2} + O(n^{-3})$ (Corollary~\ref{cor:h(d)_d=3}).
This result determines the coefficient of $O(n^{-2})$ term, denoted by $h(d)$, as $h(d=3) = {56\pi^2 \over 9}$.
By using the one-to-one correspondence of unitary estimation and port-based teleportation shown in Ref.~\cite{yoshida2024one}, this result partially solves an open problem to determine the coefficient of $O(N^{-2})$ term in $F_\mathrm{PBT}(N,d)$ raised in Refs.~\cite{christandl2021asymptotic, leditzky2022optimality} for $d=3$.
The proof of achievability is similar to the proof for port-based teleportation shown in Ref.~\cite{christandl2021asymptotic}, which shows a lower bound on the optimal fidelity of unitary estimation given by
\begin{align}
\label{eq:h(d)_upper_bound}
    h(d)\leq {\lambda_\mathrm{min}(\Omega_{d-1}) \over d}
\end{align}
for an arbitrary $d$.
The proof of the optimality is based on the finite element method, which can only be applied for $d=2,3$.
This work also shows the analytical formula for the fidelity of unitary estimation protocol shown in Ref.~\cite{kahn2007fast} for general $d$ (Corollary~\ref{cor:kahn}).
This formula gives the best upper bound on $h(d)$ compared to the previous works~\cite{christandl2021asymptotic, yang2020optimal, haah2023query}, and is asymptotically tight.

Though this work shows the tightness of the upper bound~\eqref{eq:h(d)_upper_bound} on $h(d)$ only for $d=2,3$, we conjecture the tightness of Eq.~\eqref{eq:h(d)_upper_bound} for an arbitrary $d$:

\begin{conjecture}
\label{conj:h(d)}
    \begin{align}
        h(d) = {\lambda_\mathrm{min}(\Omega_{d-1}) \over d}
    \end{align}
    holds for an arbitrary $d$.
\end{conjecture}
If this conjecture is true, we can derive the following lower bound on $h(d)$ as shown in Ref.~\cite{christandl2021asymptotic}:
\begin{align}
\label{eq:conjecture_lower_bound}
    h(d) \geq {\pi d^4 \over 8e^3},
\end{align}
which corresponds to the dashed line in Fig.~\ref{fig:comparison}.
The ratio of the best-known upper bound and lower bound on $h(d)$ is given by
\begin{align}
    &{{1\over 6} (d+1)(d-1)(3d-2)(3d-1) \over \alpha(d^2-\beta^2)^2}\nonumber\\
    &= {3\over 2\alpha}+O(d^{-1})\\
    &\approx 5.33\times 10^5+O(d^{-1}),
\end{align}
where $\alpha$ and $\beta$ are defined in Eqs.~\eqref{eq:def_alpha} and \eqref{eq:def_beta}.
If we can show Conjecture~\ref{conj:h(d)}, we can shrink this ratio to
\begin{align}
    &{{1\over 6} (d+1)(d-1)(3d-2)(3d-1) \over {\pi d^4 \over 8e^3}}\nonumber\\
    &= {12e^3\over \pi} + O(d^{-1})\\
    &\approx 7.67\times 10^1+O(d^{-1}),
\end{align}
which shows a significant improvement on the lower bound on $h(d)$ (see also Fig.~\ref{fig:comparison}).
However, we have a difficulty on extending the proof of Theorem~\ref{thm:asymptotically_optimal_unitary_estimation} to $d\geq 4$ as discussed below.

One natural way to extend the proof shown in Sec.~\ref{subsubsec:proof_lower_bound} for $d\geq 4$ is to define a triangulation $\Xi_{n,d} = \{\xi_a\}_a$ of a region $\Omega_d^{(n)}$ such that any edge in a simplex $\xi_a\in \Xi_{n,d}$ is taken from $\overline{E}'_{n,d}$.
Since the length of all the edges in $\overline{E}'_{n,d}$ is the same (${\sqrt{2} \over n}$), all simplicies $\xi_a$ should be a regular simplex with the same side length.
However, since the space filling of $\RR^{d-1}$ by a regular simplex is possible only when $d=2,3$\footnote{This is a folklore result, and we can show the impossiblity of the space filling of $\RR^d$ for $d\geq 3$ as follows.
The space filling by a polytope is possible only when the dihedral angle of the polytope is $2\pi$ divided by an integer $m$.
Since the dihedral andle of the $d$-dimensional regular simplex is $\arccos(1/d)$\cite{parks2002elementary} for $d\geq 2$, this leads to $\arccos(1/d)=2\pi/m$, i.e., $d = 1/\cos(2\pi/m)$. For $d\geq 3$, since $\arccos(1/d)>\arccos(1/2) = \pi/3$ hold, we obtain $4<m<6$, i.e., $m=5$.
However, since $1/\cos(2\pi/5) \approx 3.24$ is not an integer, the space filling by a $d$-dimensional regular simplex is impossible for $d\geq 3$.}, such a construction is impossible for $d\geq 4$.
One solution to circumvent this problem is to consider a space filling of $\Omega_d^{(n)}$ using two or more types of regular polytopes.
For instance, $\RR^3$ can be filled with the regular tetrahedron and the regular octahedron.
For such a filling, we cannot take a pairwise linear test function in the finite element method (see Appendix~\ref{appendix_sec:review_FEM}), and we need to take more general functions, e.g., higher-order polynomial functions.
However, it is not trivial whether such a generalization is possible to provide a similar analysis as shown in Eq.~\eqref{eq:lambda_FEM}, and we leave the proof of Conjecture~\ref{conj:h(d)} for a future work.

\section*{Acknowledgements}
We acknowledge fruitful discussions with Hosho Katsura and Jisho Miyazaki.
This work was supported by the MEXT
Quantum Leap Flagship Program (MEXT QLEAP) JPMXS0118069605,
JPMXS0120351339, Japan Society for the Promotion of Science (JSPS) KAKENHI Grants No. 23KJ0734 and No. 23K2164, FoPM, WINGS Program, the University of Tokyo, DAIKIN Fellowship Program, the University of Tokyo, IBM Quantum, the Special Postdoctoral Researchers Program at RIKEN, and JST ERATO Grant No. JPMJER2302, Japan

\section*{Author declarations}

The authors have no conflicts to disclose.

\section*{Data availability statement}

Data sharing is not applicable to this article as no new data were created or analyzed in this study.

\onecolumngrid
\appendix

\section{Review on the finite element method}
\label{appendix_sec:review_FEM}

In this section, we review the finite element method for the Dirichlet Laplacian problem used in the proof of Theorem~\ref{thm:asymptotically_optimal_unitary_estimation}.
For more comprehensive review, see standard textbooks, e.g., Refs.~\cite{brezis2011functional, larsson2003partial, urakawa2009geometry}.

Suppose $\Omega \subset \RR^d$ is a $d$-dimensional polytope and $\Xi = \{\xi_a\}_{a\in A}$ is a triangulation of $\Omega$, i.e., the set of simplices $\xi_i$ such that $\bigcup_{a\in A} \xi_a = \Omega$ and any two of simplices $\xi_a, \xi_b$ intersect in a common face or not at all.
We consider the following eigenvalue problem of the Dirichilet Laplacan:
\begin{align}
    \label{eq:dirichlet_laplacian_Omega}
    \begin{split}
        \Delta u &= -\lambda u,\\
        u\vert_{\partial \Omega} &= 0,\\
        u&\not\equiv 0.
    \end{split}
\end{align}
The minimum eigenvalue of the above eigenvalue problem, denoted by $\lambda_\mathrm{min}(\Omega)$, is given by the minimum principle as
\begin{align}
    \lambda_\mathrm{min}(\Omega) = \inf_{u\in H_0^1 (\Omega)} R_{\Omega}(u),
\end{align}
where $R_\Omega(u)$ is the Rayleigh quotient defined in Eq.~\eqref{eq:rayleigh} for $\Omega = \Omega_{d-1}$ and $H_0^1(\Omega)$ is the Sobolev space.
Therefore, for any $u\in H_0^1(\Omega)$ (called the test function), we obtain
\begin{align}
    \lambda_\mathrm{min}(\Omega) \leq R_\Omega(u).
\end{align}
The finite element method is to find a test function $u$ by discretizing $\Omega$ based on the triangulation $\Xi$.
We define the set of vertices $V_\Xi = \{v_i\}_{i\in I}$ by
\begin{align}
    V_\Xi\coloneqq \bigcup_{a\in A} V_{\xi_a},
\end{align}
where $V_{\xi_i}$ is the vertex set of the simplex $\xi_i$.
For $v_i\in V_\Xi$, we define a function $\psi_i: \Omega \to \RR$ such that
\begin{itemize}
    \item $\psi_i (v_j) = \delta_{i=j}$ for all $i,j \in I$
    \item $\psi_i$ is a pairwise linear function on each simplex $\xi_i\in \Xi$, i.e., there exist real constants $\alpha^a_{ij}$ and $\beta^{a}_i$ such that $\xi_i(x_1,\ldots,x_d) = \sum_{j=1}^{n} \alpha^a_{ij} x_j + \beta^{a}_{i}$ holds for all $\vec{x} = (x_1,\ldots,x_d)\in \xi_a$ and $a$.
\end{itemize}
We define the set of indices $I_0$ by
\begin{align}
    I_0\coloneqq \{i\in I\mid v_i\notin \partial \Omega\},
\end{align}
and define a real vector $\vec{u} = (u_i)_{i\in I_0}$.
Defining a pairwise linear function $\hat{\mathbf{u}}: \Omega \to \RR$ by
\begin{align}
    \hat{\mathbf{u}}(x_1,\ldots,x_d)\coloneqq \sum_{i\in I_0} u_i \psi_i(x_1,\ldots,x_d),
\end{align}
$\hat{\mathbf{u}}\in H_0^1(\Omega)$ holds for any $\vec{\mathbf{u}}$.
Therefore, we obtain
\begin{align}
    \lambda_\mathrm{min}(\Omega) \leq \inf_{\hat{\mathbf{u}}\neq 0} R_\Omega (\hat{\mathbf{u}}).
\end{align}
Defining the stiffness matrix $K(\Xi)$ and the mass matrix $M(\Xi)$ by
\begin{align}
\label{eq:def_stiffness}
    K_{ij}(\Xi) &\coloneqq \int_{\Omega} \nabla \psi_i(\vec{x}) \cdot \nabla \psi_j(\vec{x}) \dd \vec{x},\\
\label{eq:def_mass}
    M_{ij}(\Xi) &\coloneqq \int_{\Omega} \psi_i(\vec{x}) \psi_j(\vec{x}) \dd \vec{x}
\end{align}
for all $i,j\in I_0$, the Rayleigh quotient $R_\Omega(\hat{\mathbf{u}})$ is given by
\begin{align}
    R_\Omega(\hat{\mathbf{u}}) = {\vec{\mathbf{u}}^\mfT K(\Xi) \vec{\mathbf{u}} \over \vec{\mathbf{u}}^\mfT M(\Xi) \vec{\mathbf{u}}}.
\end{align}
By taking the minimum eigenvector $\vec{\textbf{u}}_\mathrm{min}$ of $M^{-1}(\Xi)K(\Xi)$ such that $M^{-1}(\Xi)K(\Xi)\vec{\textbf{u}}_\mathrm{min} = \lambda_\mathrm{min} \vec{\textbf{u}}_\mathrm{min}$ for $\lambda_\mathrm{min} = \min \eig M^{-1}(\Xi)K(\Xi)$, we obtain
\begin{align}
    R_\Omega(\hat{\mathbf{u}}_\mathrm{min}) &= \lambda_{\mathrm{min}} = \min \eig M^{-1}(\Xi)K(\Xi),
\end{align}
i.e.,
\begin{align}
    \lambda_\mathrm{min}(\Omega) \leq \min \eig M^{-1}(\Xi)K(\Xi)
\end{align}
holds.
In conclusion, we obtain the following theorem:
\begin{theorem}
\label{thm:FEM}
    The minimum eigenvalue $\lambda_\mathrm{min}(\Omega)$ of the Dirichlet Laplacian problem~\eqref{eq:dirichlet_laplacian_Omega} satisfies
    \begin{align}
        \lambda_\mathrm{min}(\Omega)\leq \min \eig M^{-1}(\Xi)K(\Xi),
    \end{align}
    where $K(\Xi)$ and $M(\Xi)$ are the stiffness matrix and the mass matrix defined in Eqs.~\eqref{eq:def_stiffness} and \eqref{eq:def_mass}.
\end{theorem}

\section{Derivation of analytical formula for Kahn's lower bound on the unitary estimation fidelity}
\label{appendix_sec:kahn}
\subsection{Review of Kahn's derivation}
\label{appendix_subsec:kahn_derivation}
As shown in Sec.~\ref{subsubsec:proof_upper_bound}, for a smooth function $u:\Omega_{d-1} \to \RR$ such that $u\vert_{\partial\Omega_{d-1}}=0$, we obtain
\begin{align}
    h_{n,d}\leq {R_{\Omega_{d-1}}(u)\over d}+O(n^{-1}),
\end{align}
where $R_{\Omega_{d-1}}(u)$ is the Rayleigh quotient defined in Eq.~\eqref{eq:rayleigh}.
We define $u$ by
\begin{align}
    u(x_1,\ldots,x_{d})\coloneqq x_d \prod_{i=1}^{d-1} (x_i - x_{i+1}),
\end{align}
which satisfies $u\vert_{\partial \Omega_{d-1}} = 0$.
To evaluate the integral in Eq.~\eqref{eq:rayleigh}, we introduce a new coordinate $(y_1,\ldots,y_d)$ defined by
\begin{align}
    y_i\coloneqq
    \begin{cases}
        x_i-x_{i+1} & (i\in[d-1])\\
        x_d & (i=d)
    \end{cases},
\end{align}
which satisfies
\begin{align}
    (x_1,\ldots,x_d) \in \Omega_{d-1} \Leftrightarrow (y_1,\ldots,y_d)\in S_{d-1}.
\end{align}
Then, the function $u$ and $\nabla u$ is expressed in the new coordinate as
\begin{align}
    u(x_1,\ldots,x_d) &= v(y_1,\ldots,y_d)\coloneqq \prod_{i=1}^{d} y_i,\\
    {\partial u(x_1,\ldots,x_d) \over \partial x_i} &=
    \begin{cases}
        {\partial v(y_1,\ldots,y_d) \over \partial y_1}& (i=1)\\
        {\partial v(y_1,\ldots,y_d) \over \partial y_i} - {\partial v(y_1,\ldots,y_d) \over \partial y_{i-1}} & (i\geq 2)\\
    \end{cases}\\
    &= \begin{cases}
        y_{\{1\}}& (i=1)\\
        y_{\{i\}}-y_{\{i-1\}} & (i\geq 2)\\
    \end{cases}.
\end{align}
Thus, $R_{\Omega_{d-1}}(u)$ is evaluated as
\begin{align}
    R_{\Omega_{d-1}}(u)
    &= {\int_{\Omega_{d-1}} \sum_{i=1}^{d} \left[{\partial u (x_1,\ldots,x_d) \over \partial x_i}\right]^2 - {1\over d}\left[\sum_{i=1}^{d}{\partial u (x_1,\ldots,x_d) \over \partial x_i}\right]^2 \dd \vec{x} \over \int_{\Omega_{d-1}} [u(x_1,\ldots,x_d]^2 \dd \vec{x}}\\
    &= {\int_{S_{d-1}} \left[(y_{\{1\}})^2+\sum_{i=2}^{d}(y_{\{i\}} - y_{\{i-1\}})^2 - {1\over d} (y_{\{d\}})^2\right] \dd \vec{y} \over \int_{S_{d-1}} \prod_{i=1}^{d} y_i^2 \dd \vec{y}}\\
    &= {\int_{S_{d-1}} \left[2\sum_{i=1}^{d} (y_{\{i\}})^2-2\sum_{i=2}^{d}y_{\{i-1\}} y_{\{i\}}-\left(1+{1\over d}\right) (y_{\{d\}})^2\right] \dd \vec{y} \over \int_{S_{d-1}} \prod_{i=1}^{d} y_i^2 \dd \vec{y}}\\
    &= {2d(A_d-B_d) - (d+1) C_d\over d D_d},
\end{align}
i.e., we obtain
\begin{align}
    h_{n,d} \leq {2d(A_d-B_d) - (d+1) C_d\over d^2 D_d}+O(n^{-1}).
\end{align}

\subsection{Analytical solution for the integrals in Kahn's upper lower bound}
\label{appendix_subsec:kahn_solution}
This section shows the analytical formulae for $A_d, B_d, C_d, D_d$ shown in Eq.~\eqref{eq:kahn_integral}:
\begin{lemma}
\label{lem:kahn_integral}
    \begin{align}
        A_d&= {d(d+1)(2d+1)(3d-1)(3d-2) \over 12}D_d,\\
        B_d&= {d(d-1)(d+1)(3d-1)(3d-2)\over 12}D_d,\\
        C_d&= {d^2(3d-2)(3d-1) \over 2}D_d,\\
        D_d&= {2^d \sqrt{5}\over (d!)^3 (3d-1)!}.
    \end{align}
\end{lemma}
\begin{proof}
We focus on the recursive structure on $S_{d-1}$ given by
\begin{align}
    S_{d-1}
    &= \{((1-dx_d)\vec{x}',x_{d}) \mid \vec{x}'\in S_{d-2}, 0\leq x_d\leq 1/d\}\\
    &= \{((1-dx_d-(d-1)x_{d-1})\vec{x}'',x_{d-1},x_{d}) \mid \vec{x}''\in S_{d-3}, 0\leq x_d\leq 1/d, 0\leq x_{d-1}\leq (1-dx_d)/(d-1)\}.
\end{align}
Then, $A_d$, $B_d$, $C_d$ and $D_d$ satisfy the following recursive relation:
\begin{align}
    A_d
    &= \int_{0}^{1/d} \dd x_d x_d^2 (1-dx_d)^{3(d-2)} A_{d-1} + C_d\\
    &= {2\over d^3(3d-5)(3d-4)(3d-3)} A_{d-1}+C_d \quad \text{for } d\geq 3,\\
    B_d
    &= \int_0^{1/d} \dd x_d x_d^2 (1-dx_d)^{3(d-2)} B_{d-1}\nonumber\\
    &\quad + \int_{0}^{1/d} \dd x_d x_d \int_{0}^{(1-dx_d)/(d-1)} \dd x_{d-1} x_{d-1} [1-dx_d-(d-1)x_{d-1}]^{2(d-2)+d-3} D_{d-2}\\
    &= {2\over d^3(3d-5)(3d-4)(3d-3)} B_{d-1} + {1\over d^2(d-1)^2(3d-6)(3d-5)(3d-4)(3d-3)} D_{d-2} \quad \text{for } d\geq 4,\\
    C_d
    &= \int_{0}^{1/d} \dd x_d (1-dx_d)^{2(d-1)+d-2} D_{d-1}\\
    &= {1\over d(3d-3)} D_{d-1} \quad \text{for } d\geq 3,\\
    D_d
    &= \int_0^{1/d} \dd x_d x_d^2 (1-dx_d)^{2(d-1)+d-2} D_{d-1}\\
    &= {2\over d^3 (3d-3)(3d-2)(3d-1)} D_{d-1} \quad \text{for } d\geq 3.
\end{align}
Since
\begin{align}
    A_2 = {5\sqrt{5}\over 24},
    \quad B_2 = {\sqrt{5}\over 24},
    \quad B_3 = {\sqrt{5}\over 9720},
    \quad C_2 = {\sqrt{5} \over 6},
    \quad
    D_2 = {\sqrt{5}\over 240}
\end{align}
hold, we obtain
\begin{align}
    A_d&= {d(d+1)(2d+1)(3d-1)(3d-2) \over 12}D_d,\\
    B_d&= {d(d-1)(d+1)(3d-1)(3d-2)\over 12}D_d,\\
    C_d&= {d^2(3d-2)(3d-1) \over 2}D_d,\\
    D_d&= {2^d \sqrt{5}\over (d!)^3 (3d-1)!}.
\end{align}
\end{proof}

\section*{References}

\twocolumngrid

\bibliography{main}

\end{document}